\documentclass{llncs}
\usepackage{amsmath,amssymb}
\usepackage{tabls}
\usepackage{graphicx}
\usepackage{wrapfig}
\usepackage{array}
\usepackage[algoruled,linesnumbered,algo2e,vlined]{algorithm2e}
\usepackage{url}
\usepackage{amssymb}
\usepackage{multirow}
\usepackage{multicol}

\usepackage[final,mode=multiuser]{fixme}
\fxuseenvlayout{colorsig}
\fxusetargetlayout{color}
\FXRegisterAuthor{yn}{ayn}{YN}
\FXRegisterAuthor{ti}{ati}{TI}
\FXRegisterAuthor{si}{asi}{SI}
\FXRegisterAuthor{hb}{ahb}{HB}
\FXRegisterAuthor{mt}{amt}{MT}
\fxsetface{env}{}

\newcommand{\derive}{\mathit{val}}

\newlength\savedwidth

\newcommand{\LF}{\mathit{LF}}
\newcommand{\kouho}{\mathit{LFCand}}
\newcommand{\suffix}{\mathit{Suffix}}

\newcommand{\FirstMismatch}{\mathit{FM}}
\newcommand{\lcp}{\mathit{lcp}}

\title{
  Efficient Lyndon factorization of grammar compressed text
}

 \author{
   Tomohiro I\inst{1,2}
   \and
   Yuto Nakashima\inst{1}
   \and
   Shunsuke Inenaga\inst{1}
   \and 
   Hideo Bannai\inst{1}
   \and \\
   Masayuki Takeda\inst{1}
 }
 \institute{
   Department of Informatics, Kyushu University, Japan \\
   \email{\{tomohiro.i, yuto.nakashima, inenaga, bannai, takeda\}@inf.kyushu-u.ac.jp}
   \and
   Japan Society for the Promotion of Science (JSPS)
 }

\begin{document}
\maketitle

\begin{abstract}
  We present an algorithm for computing the Lyndon factorization of a
  string that is given in grammar compressed form, namely, a Straight
  Line Program (SLP).
  The algorithm runs in $O(n^4 + mn^3h)$ time and $O(n^2)$ space, where
  $m$ is the size of the Lyndon factorization, $n$ is the size of
  the SLP, and $h$ is the height of the derivation tree of the SLP.
  Since the length of the decompressed string can be exponentially large 
  w.r.t. $n, m$ and $h$,
  our result is the first polynomial time solution 
  when the string is given as SLP.
\end{abstract}
\section{Introduction}

\sinote{added}{%
\emph{Compressed string processing} (\emph{CSP}) is a task of 
processing compressed string data without explicit decompression.
As any method that first decompresses the data requires time and space
dependent on the decompressed size of the data,
CSP without explicit decompression has been gaining importance due
to the ever increasing amount of data produced and stored.
A number of efficient CSP algorithms have been proposed,
e.g., see~\cite{hermelin09:_unified_algor_accel_edit_distan,yamamoto11:_faster_subseq_dont_care_patter,goto13:_fast_slp,gawrychowski11:_LZ_comp_str_fast_,gawrychowski11:_optimal_LZW_,gawrychowski12:_faster_algor_comput_edit_distan}.
In this paper, we present new CSP algorithms that compute 
the \emph{Lyndon factorization} of strings.
}%

A string $\ell$ is said to be a \emph{Lyndon word}
if $\ell$ is lexicographically smallest among
its circular permutations of characters of $\ell$.
For example, $\mathtt{aab}$ is a Lyndon word,
but its circular permutations $\mathtt{aba}$ and $\mathtt{baa}$ are not.
Lyndon words have various and important applications in, 
e.g., 
musicology~\cite{Chemillier04:_periodic_musical_Lyndon},
bioinformatics~\cite{Delgrange04:_STAR_tandem_repeat},
approximation algorithm~\cite{Mucha13:_lyndon_superstring},
string matching~\cite{crochemore91:_two-way_matching,breslauer11:_simple_real-time,neuburger11:_succinct_2D},
word combinatorics~\cite{fredricksen78:_necklaces_de_Bruijn,Provencal11},
and free Lie algebras~\cite{lyndon54:_burnside}.

The \emph{Lyndon factorization} (a.k.a. \emph{standard factorization}) of a string $w$,
denoted $\LF(w)$,
is a unique sequence of Lyndon words such that 
the concatenation of the Lyndon words gives $w$ and
the Lyndon words in the sequence are lexicographically 
non-increasing~\cite{ChenFL58:_lyndon_factorization_}.
Lyndon factorizations are used in
a bijective variant of Burrows-Wheeler 
transform~\cite{kufleitner09:_bijective_BWT,Gil12:_bijective_sorting_}
and a digital geometry algorithm~\cite{BrlekLPR09}.
Duval~\cite{Duval83:_facrorizing_words_} proposed
an elegant on-line algorithm to compute 
$\LF(w)$ of a given string $w$ of length $N$ in $O(N)$ time.
Efficient parallel algorithms to compute 
the Lyndon factorization are also known~\cite{ApostolicoC95,DaykinIS94}.

We present a new CSP
algorithm which computes the Lyndon factorization $\LF(w)$ of a string $w$, 
when $w$ is given in a \emph{grammar-compressed form}.
Let $m$ be the number of factors in $\LF(w)$.
Our first algorithm computes $\LF(w)$ in $O(n^4 + mn^3h)$ time and $O(n^2)$ space,
where $n$ is the size of a given \emph{straight-line program} (\emph{SLP}),
which is a context-free grammar 
in Chomsky normal form that derives only $w$, and $h$ is the height of the derivation tree of the SLP.
Since the decompressed string length $|w| = N$ can be exponentially 
large w.r.t. $n, m$ and $h$,
our $O(n^4 + mn^3h)$ solution can be efficient for highly compressive strings.
%
%
%
%
%

\section{Preliminaries}

\subsection{Strings and model of computation}

Let $\Sigma$ be a finite {\em alphabet}.
An element of $\Sigma^*$ is called a {\em string}.
The length of a string $w$ is denoted by $|w|$. 
The empty string $\varepsilon$ is a string of length 0,
namely, $|\varepsilon| = 0$.
Let $\Sigma^+$ be the set of non-empty strings,
i.e., $\Sigma^+ = \Sigma^* - \{\varepsilon\}$.
For a string $w = xyz$, $x$, $y$ and $z$ are called
a \emph{prefix}, \emph{substring}, and \emph{suffix} of $w$, respectively.
A prefix $x$ of $w$ is called a \emph{proper prefix} of $w$
if $x \neq w$, i.e., $x$ is shorter than $w$.
The set of suffixes of $w$ is denoted by $\suffix(w)$.
The $i$-th character of a string $w$ is denoted by $w[i]$, where $1 \leq i \leq |w|$.
For a string $w$ and two integers $1 \leq i \leq j \leq |w|$, 
let $w[i..j]$ denote the substring of $w$ that begins at position $i$ and ends at
position $j$. For convenience, let $w[i..j] = \varepsilon$ when $i > j$.
For any string $w$ 
let $w^1 = w$, and for any integer $k > 2$ let $w^k = ww^{k-1}$,
i.e., $w^k$ is a $k$-time repetition of $w$.

A positive integer $p$ is said to be a \emph{period} of 
a string $w$ if $w[i] = w[i+p]$ for all $1 \leq i \leq |w|-p$.
Let $w$ be any string and $q$ be its smallest period.
If $p$ is a period of a string $w$
\hbnote*{added}{%
  such that $p < |w|$,
}%
then the positive integer $|w|-p$ is said to be a \emph{border} of $w$.
If $w$ has no borders, 
then $w$ is said to be \emph{border-free}.

If character $a \in \Sigma$ is lexicographically smaller than 
another character $b \in \Sigma$,
then we write $a \prec b$.
For any non-empty strings $x, y \in \Sigma^+$,
let $\lcp(x, y)$ be the length of the longest common prefix of $x$ and $y$.
We denote $x \prec y$, if either of the following conditions holds:
$x[\lcp(x, y)+1] \prec y[\lcp(x, y) + 1]$, or $x$ is a proper prefix of $y$.
For a set $S \subseteq \Sigma^+$ of non-empty strings,
let $\min_\prec S$ denote the lexicographically smallest string in $S$.

Our model of computation is the word RAM:
We shall assume that the computer word size is at least $\lceil \log_2 |w| \rceil$, 
and hence, standard operations on
values representing lengths and positions of string $w$
can be manipulated in constant time.
Space complexities will be determined by the number of computer words (not bits).

\subsection{Lyndon words and Lyndon factorization of strings}

Two strings $x$ and $y$ are said to be \emph{conjugate},
if there exist strings $u$ and $v$ such that $x = uv$ and $y = vu$.
A string $w$ is said to be a \emph{Lyndon word}, 
if $w$ is lexicographically strictly smaller than all of its conjugates of $w$.
Namely, $w$ is a Lyndon word, 
if for any factorization $w = uv$, it holds that $uv \prec vu$.
It is known that any Lyndon word is border-free.

\begin{definition}[\cite{ChenFL58:_lyndon_factorization_}]
The \emph{Lyndon factorization} of a string $w$,
denoted $\LF(w)$,
is the factorization $\ell_1^{p_1} \cdots \ell_m^{p_m}$ of $w$,
such that each $\ell_i \in \Sigma^+$ is a Lyndon word,
$p_{i} \geq 1$, and $\ell_{i} \succ \ell_{i+1}$ for all $1 \leq i < m$.
\end{definition}

It is known that the Lyndon factorization is unique for each string $w$,
and it was shown by Duval~\cite{Duval83:_facrorizing_words_}
that the Lyndon factorization can be computed in $O(N)$ time, where $N = |w|$.

$\LF(w)$ can be represented by the sequence $(|\ell_1|, p_1), \ldots, (|\ell_m|, p_m)$ of 
integer pairs,
where each pair $(|\ell_i|, p_i)$ represents the $i$-th Lyndon factor 
$\ell_i^{p_i}$ of $w$.
Note that this representation requires $O(m)$ space.

\subsection{Straight line programs}
\label{sec:slp}

A {\em straight line program} ({\em SLP}) is a set of productions 
$\mathcal S = \{ X_1 \rightarrow expr_1, X_2 \rightarrow expr_2, \ldots, X_n \rightarrow expr_n\}$,
where each $X_i$ is a variable and each $expr_i$ is an expression, where
$expr_i = a$ ($a\in\Sigma$), or $expr_i = X_{\ell(i)} X_{r(i)}$~($i > \ell(i),r(i)$).
It is essentially a context free grammar in Chomsky normal form, that derives a single string.
Let $\derive(X_i)$ represent the string derived from variable $X_i$.
To ease notation, we sometimes associate $\derive(X_i)$ with $X_i$ and
denote $|\derive(X_i)|$ as $|X_i|$,
and $\derive(X_i)[u..v]$ as $X_i[u..v]$ for $1 \leq u \leq v \leq |X_i|$.
An SLP $\mathcal{S}$ {\em represents} the string $w = \derive(X_n)$.
The \emph{size} of the program $\mathcal{S}$ is the number $n$ of
productions in $\mathcal{S}$.
Let $N$ be the length of the string represented by SLP $\mathcal{S}$,
i.e., $N = |w|$. 
Then $N$ can be as large as $2^{n-1}$. 

The derivation tree of SLP $\mathcal{S}$ is a labeled
ordered binary tree where each internal node is labeled with a
non-terminal variable in $\{X_1,\ldots,X_n\}$, and each leaf is labeled with a terminal character in $\Sigma$.
The root node has label $X_n$.
An example of the derivation tree of an SLP is shown in Fig.~\ref{fig:SLP}.

\begin{figure}[tb]
 \centerline{\includegraphics[width=0.5\textwidth]{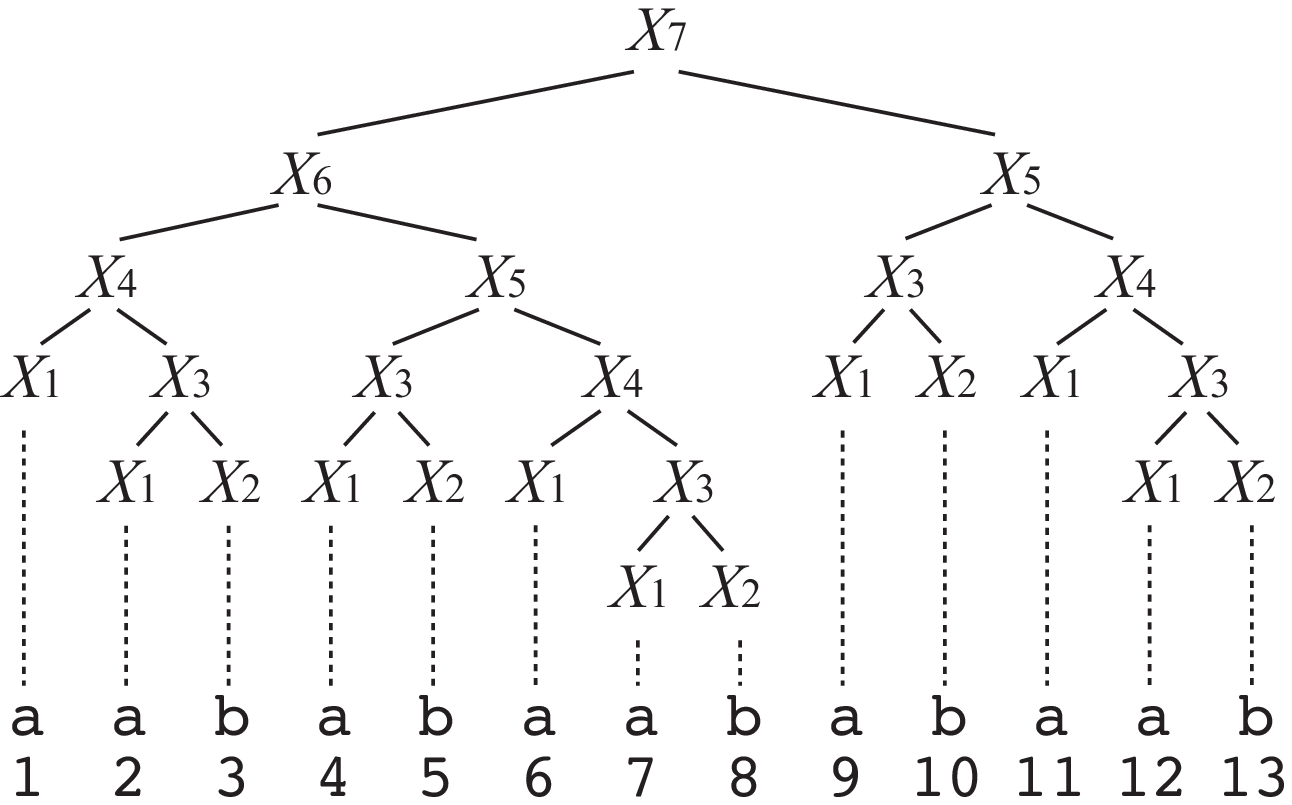}}
 \caption{
   The derivation tree of
   SLP $\mathcal S = \{ X_1 \rightarrow \mathtt{a}$, $X_2 \rightarrow \mathtt{b}$, $X_3 \rightarrow X_1X_2$,
   $X_4 \rightarrow X_1X_3$, $X_5 \rightarrow X_3X_4$, $X_6
   \rightarrow X_4X_5$, $X_7 \rightarrow X_6X_5 \}$,
   representing string $S = \derive(X_7) = \mathtt{aababaababaab}$.
 }
 \label{fig:SLP}
\end{figure}

%
%
%
%
%
%
%
%
%
%
%
%
%
%
%
%
%
%
%
%
%
%
%
%
%

%
%
%
%
%
%
%
%
%
%
%
%
%
%
%
%
%

%
%
%
%
%
%
%
%
%
%
%

\section{Computing Lyndon factorization from SLP}

In this section, we show how, given an SLP $\mathcal{S}$ of $n$ productions
representing string $w$,
we can compute $\LF(w)$ of size $m$ in $O(n^4 + mn^3h)$ time.
We will make use of the following known results:
\begin{lemma}[\cite{Duval83:_facrorizing_words_}] \label{lem:smallest_suffix}
For any string $w$, let $\LF(w) = \ell_1^{p_1}, \ldots, \ell_m^{p_m}$.
Then, $\ell_m = \min_\prec \suffix(w)$, i.e., 
$\ell_m$ is the lexicographically smallest suffix of $w$.
\end{lemma}

\begin{lemma}[\cite{lifshits:DSP:2006:798}] \label{lem:substring_SLP}
Given an SLP $\mathcal{S}$ of size $n$ representing a string $w$ of length $N$,
and two integers $1 \leq i \leq j \leq N$,
we can compute in $O(n)$ time another SLP of size $O(n)$
representing the substring $w[i..j]$.
\end{lemma}

\begin{lemma}[\cite{lifshits:DSP:2006:798}] \label{lem:SLP_period}
Given an SLP $\mathcal{S}$ of size $n$ representing a string $w$ of length $N$,
we can compute the shortest period of $w$ in $O(n^3 \log N)$ time and $O(n^2)$ space.
\end{lemma}

For any non-empty string $w \in \Sigma^+$,
let $\kouho(w) = \{x \mid x \in \suffix(w), \exists y \in \Sigma^+ \mbox{ s.t. } xy = \min_\prec \suffix(wy)\}$.
Intuitively, $\kouho(w)$ is the set of suffixes of $w$
which are a prefix of the lexicographically 
smallest suffix of string $wy$, for some non-empty string $y \in \Sigma^+$.

The following lemma may be almost trivial,
but will play a central role in our algorithm.
\begin{lemma} \label{lem:shorter_kouho_prefix}
For any two strings $u, v \in \kouho(w)$ with $|u| < |v|$,
$u$ is a prefix of $v$.
\end{lemma}

\begin{proof}
If $v[1..|u|] \prec u$,
then for any non-empty string $y$, $vy \prec uy$.
However, this contradicts that $u \in \kouho(w)$.
If $v[1..|u|] \succ u$,
then for any non-empty string $y$, $vy \succ uy$.
However, this contradicts that $v \in \kouho(w)$.
Hence we have $v[1..|u|] = u$.
\qed 
\end{proof}

\begin{lemma} \label{lem:LF_shortest_kouho}
For any string $w$, let $\ell = \min_\prec \suffix(w)$.
Then, the shortest string of $\kouho(w)$ is $\ell^{p}$,
where $p \geq 1$ is the maximum integer
such that $\ell^p$ is a suffix of $w$. 
\end{lemma}
\begin{proof}
For any string $x \in \kouho(w)$,
and any non-empty string $y$,
$xy = \min_\prec\suffix(wy)$
holds only if $y \succ \ell$. 

Firstly, we compare $\ell^p$ with the suffixes $s$ of $w$
shorter than $\ell^p$,
and show that $\ell^py \prec s y$ holds for \emph{any} $y \succ \ell$.
Such suffixes $s$ are divided into two groups:
(1) If $s$ is of form $\ell^k$ for any integer $1 \leq k < p$,
then $\ell^{p}y \prec \ell^ky = sy \prec y$ holds for any $y \succ \ell$;
(2) If $s$ is not of form $\ell^k$, then since $\ell$ is border-free,
$\ell$ is not a prefix of $s$, and $s$ is not a prefix of $\ell$, either.
Thus $\ell^p \prec s$ holds, implying that $\ell^py \prec sy$ for any $y \succ \ell$.

\sinote{added}{%
Secondly, we compare $\ell^p$ with the suffixes $t$ of $w$
longer than $\ell^p$,
and show that $\ell^py \prec t y$ holds for \emph{some} $y \succ \ell$.
By Lemma~\ref{lem:shorter_kouho_prefix},
$t = \ell^q u$ holds, where
$q \geq p$ is the maximum integer such that $\ell^q$ is a prefix of $t$,
and $u \in \Sigma^+$.
By definition, $\ell \prec u$ and $\ell$ is not a prefix of $u$.
Choosing $y = \ell^{q-p} u^\prime$ with $u^\prime \prec u$,
we have $\ell^p y = \ell^{q} u^\prime \prec \ell^q u = t \prec ty$.
Hence, $\ell^p \in \kouho(w)$ and no shorter strings exist in $\kouho(w)$.
}%
\qed
\end{proof}

By Lemma~\ref{lem:smallest_suffix} and Lemma~\ref{lem:LF_shortest_kouho},
computing the last Lyndon factor $\ell_m^{p_m}$ of $w = \derive(X_n)$
reduces to computing $\kouho(X_n)$ for the last variable $X_n$.
In what follows, we propose a dynamic programming algorithm to compute 
$\kouho(X_i)$ for each variable.
Firstly we show the number of strings in $\kouho(X_i)$ is $O(\log N)$,
where $N = |\derive(X_n)| = |w|$.

\begin{lemma} \label{lem:kouho_twice_longer}
For any string $w$,
let $s_j$ be the $j$th shortest string of $\kouho(w)$.
Then, $|s_{j+1}| > 2|s_j|$ for any $1 \leq j < |\kouho(w)|$.
\end{lemma}
\begin{proof}
Let $\ell = \min_\prec \suffix(w)$,
and $y$ any string such that $y \succ \ell$.
It follows from Lemma~\ref{lem:shorter_kouho_prefix}
that $\ell$ is a prefix of any string $s_j \in \kouho(w)$,
and hence $s_j \prec y$ holds.

Assume on the contrary that $|s_{j+1}| \leq 2|s_j|$.
If $|s_{j+1}| = 2|s_j|$, i.e., $s_{j+1} = s_j s_j$, then $s_{j+1} y = s_j s_j y \prec s_j y$ holds, 
but this contradicts that $s_j \in \kouho(w)$.
Hence $s_{j+1} \neq s_j s_j$.
If $|s_{j+1}| < 2|s_{j}|$,
by Lemma~\ref{lem:shorter_kouho_prefix}, 
$s_j$ is a prefix of $s_{j+1}$, and therefore 
$s_{j}$ has a period $q$ such that $s_{j+1} = u^k v$ and $s_{j} = u^{k-1} v$, 
where $u = s_j[1..q]$, $k \geq 1$ is an integer,
and $v$ is a proper prefix of $u$.
There are two cases to consider:
\sinote{modified}{%
(1) If $uvy \prec vy$, then $u^kvy \prec u^{k-1}vy = s_{j}y$. 
(2) If $vy \prec uvy$, then $vy \prec uvy \prec u^2vy \prec \cdots \prec u^{k-1}vy = s_{j}y$.
It means that $\min_\prec \{u^kvy, vy\} \prec s_{j}y$ for any $y \succ \ell$, however, this contradicts that $s_{j} \in \kouho(w)$.
Hence $|s_{j+1}| > 2|s_{j}|$ holds.
}%
\qed
\end{proof}
Since $s_{j}$ is a suffix of $s_{j+1}$, 
it follows from Lemma~\ref{lem:shorter_kouho_prefix}
and Lemma~\ref{lem:kouho_twice_longer} that
$s_{j+1} = s_{j}ts_{j}$ with some non-empty string $t \in \Sigma^+$.
This also implies that the number of strings in $\kouho(w)$ is $O(\log N)$,
where $N$ is the length of $w$.
By identifying each suffix of $\kouho(X_i)$ with its length,
and using Lemma~\ref{lem:kouho_twice_longer},
$\kouho(X_i)$ for all variables can be stored in a total of $O(n \log N)$ space.

For any two variables $X_i, X_j$ of an SLP $\mathcal{S}$
and a positive integer $k$ satisfying $|X_i| \geq k + |X_j| -1$,
consider the $\FirstMismatch$ function such that
$\FirstMismatch(X_i, X_j, k) = \lcp(\derive(X_i)[k..|X_i|], \derive(X_j))$,
i.e., it returns the length of the lcp of the suffix of $\derive(X_i)$
starting at position $k$ and $X_j$.

\begin{lemma}[\cite{MasamichiCPM97,lifshits07:_proces_compr_texts}] \label{lem:lcp_slp}
We can preprocess a given SLP $\mathcal{S}$ of size $n$ 
in $O(n^3)$ time and $O(n^2)$ space so that
$\FirstMismatch(X_i, X_j, k)$ can be answered in $O(n^2)$ time.
\end{lemma}
For each variable $X_i$
we store the length $|X_i|$ of the string derived by $X_i$.
It requires a total of $O(n)$ space
for all $1 \leq i \leq n$,
and can be computed in a total of $O(n)$ time 
by a simple dynamic programming algorithm.
Given a position $j$ of the uncompressed string $w$ of length $N$,
i.e., $1 \leq j \leq N$,
we can retrieve the $j$th character $w[j]$ in $O(n)$ time
by a simple binary search on the derivation tree of $X_n$ using the lengths stored in the variables.
Hence, we can lexicographically compare $\derive(X_i)[k..|X_i|]$ and $\derive(X_j)$
in $O(n^2)$ time, after $O(n^3)$-time preprocessing.

The following lemma shows a dynamic programming approach
to compute $\kouho(X_i)$ for each variable $X_i$.
\sinote*{added}{%
We will mean by a sorted list of $\kouho(X_i)$
the list of the elements of $\kouho(X_i)$ sorted 
in increasing order of length.
}%

\begin{lemma} \label{lem:kouho_O(n^3)}
Let $X_i = X_\ell X_r$ be any production of 
a given SLP $\mathcal{S}$ of size $n$.
Provided that sorted lists for 
$\kouho(X_\ell)$ and $\kouho(X_r)$
are already computed, 
a sorted list for $\kouho(X_i)$ can be computed in $O(n^3)$ time and $O(n^2)$ space.
\end{lemma}

\begin{proof}
Let $D_i$ be a sorted list of the suffixes of $X_i$
that are candidates of elements of $\kouho(X_i)$.
We initially set $D_i \leftarrow \kouho(X_r)$.

We process the elements of $\kouho(X_\ell)$ in increasing order of length.
Let $s$ be any string in $\kouho(X_\ell)$,
and $d$ the longest string in $D_i$.
\sinote*{added}{%
Since any string of $\kouho(X_r)$ is a prefix of $d$ by Lemma~\ref{lem:shorter_kouho_prefix},
in order to compute $\kouho(X_i)$
it suffices to lexicographically compare $s \cdot \derive(X_r)$ and $d$.
}%
Let $h = \lcp(s \cdot \derive(X_r), d)$).
See also Fig.~\ref{fig:kouho_0}.

\begin{figure}[tb]
  \centerline{\includegraphics[width=0.8\textwidth]{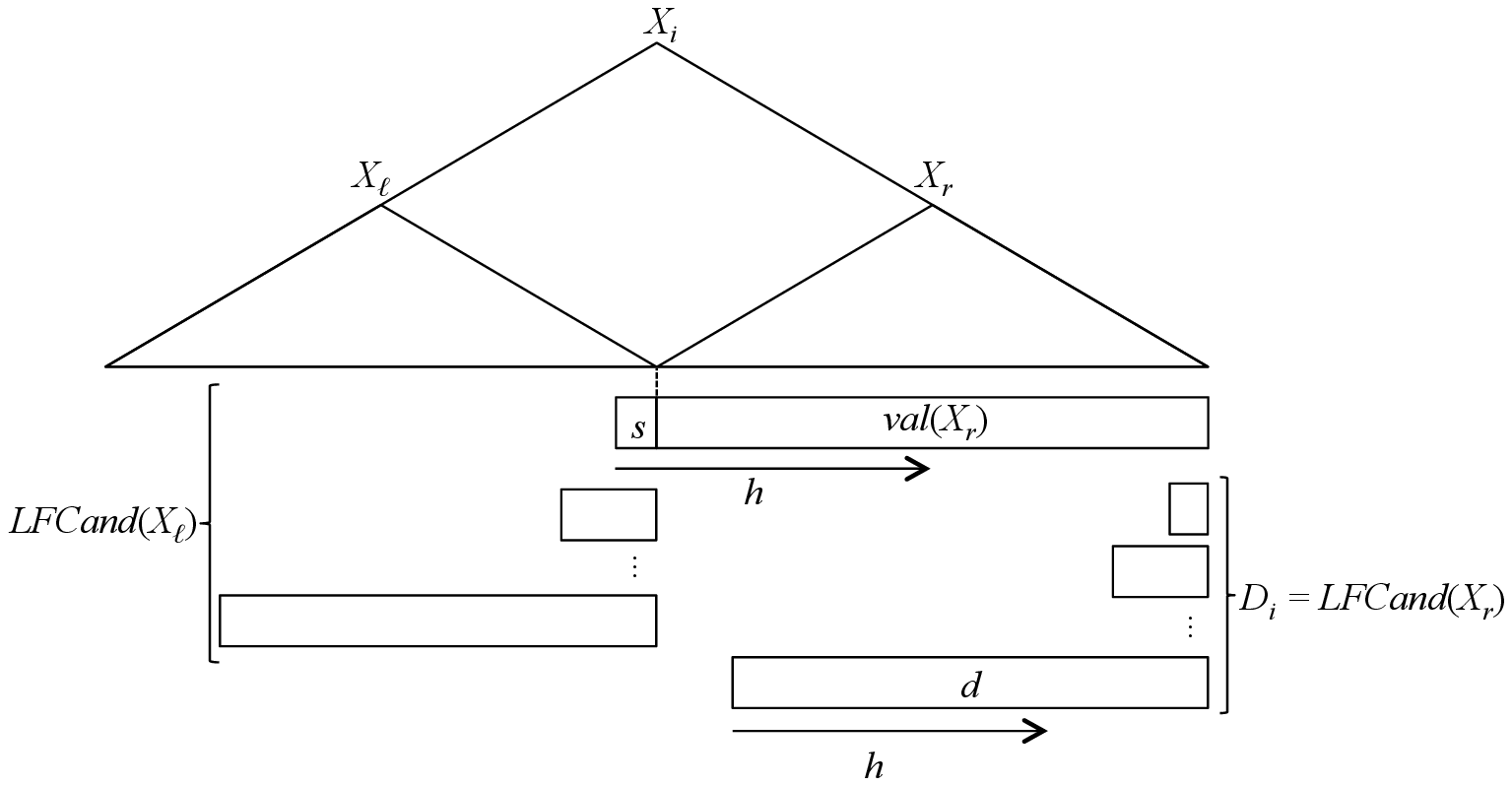}}
  \caption{Lemma~\ref{lem:kouho_O(n^3)}: Initially $D_i = \kouho(X_r)$ and $h = s \cdot \derive(X_\ell)$ with $s$ being the shortest string of $\kouho(X_\ell)$.
  }
  \label{fig:kouho_0}
\end{figure}

\begin{itemize}
\item If $(s \cdot \derive(X_r))[h+1] \prec d[h+1]$, then $s \cdot \derive(X_r) \prec d$.
Since any string in $D_i$ is a prefix of $d$ by Lemma~\ref{lem:shorter_kouho_prefix},
we observe that any element in $D_i$ that is longer than $h$ cannot be 
an element of $\kouho(X_i)$.
Hence we delete any element of $D_i$ that is longer than $h$ from $D_i$,
then add $s \cdot \derive(X_r)$ to $D_i$,
and update $d \leftarrow s \cdot \derive(X_r)$.
See also Fig.~\ref{fig:kouho_1}.
\begin{figure}[tbhp]
  \centerline{\includegraphics[width=0.7\textwidth]{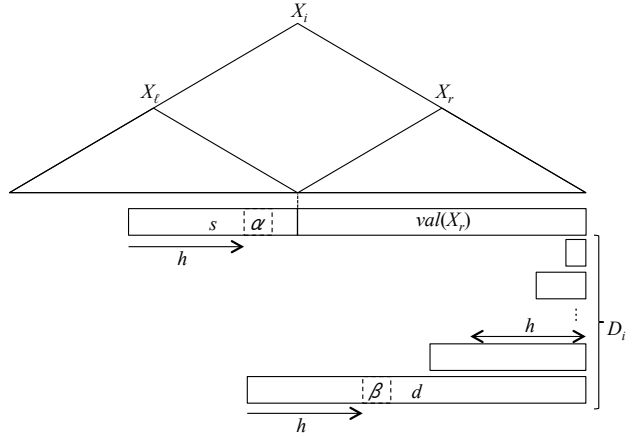}}
  \caption{Lemma~\ref{lem:kouho_O(n^3)}: Case where $(s \cdot \derive(X_r))[h+1] = \alpha \prec d[h+1] = \beta$. $d$ and any string in $D_i$ that is longer than $h$ are deleted from $D_i$. Then $s \cdot \derive(X_r)$ becomes the longest candidate in $D_i$.
  }
  \label{fig:kouho_1}
\end{figure}

\item If $(s \cdot \derive(X_r))[h+1] \succ d[h+1]$, then $s \cdot \derive(X_r) \succ d$.
Since $s \cdot \derive(X_r)$ cannot be an element of $\kouho(X_i)$, 
in this case neither $D_i$ nor $d$ is updated.
See also Fig.~\ref{fig:kouho_2}.
\begin{figure}[tbhp]
  \centerline{\includegraphics[width=0.7\textwidth]{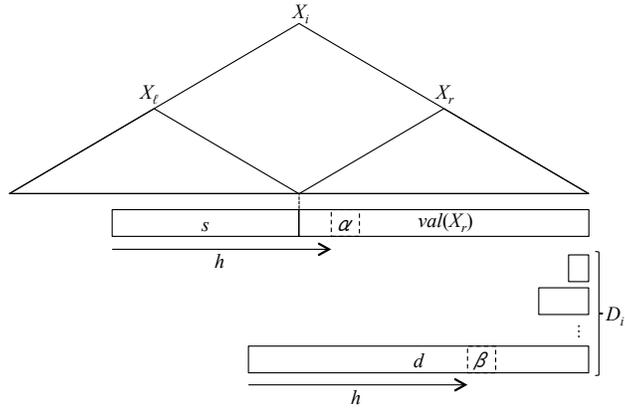}}
  \caption{Lemma~\ref{lem:kouho_O(n^3)}: Case where $(s \cdot \derive(X_r))[h+1] = \alpha \succ d[h+1] = \beta$. There are no updates on $D_i$.
  }
  \label{fig:kouho_2}
\end{figure}

\item If $h = |d|$, i.e., $d$ is a prefix of $s \cdot \derive(X_r)$,
then there are two sub-cases:
\begin{itemize}
\item If $|s \cdot \derive(X_r)| \leq 2|d|$,
$d$ has a period $q$ such that $s \cdot \derive(X_r) = u^k v$ and $d = u^{k-1} v$, 
where $u = d[1..q]$, $k \geq 1$ is an integer,
and $v$ is a proper prefix of $u$.
By similar arguments to Lemma~\ref{lem:kouho_twice_longer},
we observe that $d$ cannot be a member of $\kouho(X_i)$
while $s \cdot \derive(X_r)$ may be a member of $\kouho(X_i)$.
Thus we add $s \cdot \derive(X_r)$ to $D_i$, delete $d$ from $D_i$,
and update $d \leftarrow s \cdot \derive(X_r)$.
See also Fig.~\ref{fig:kouho_3_1}.
\begin{figure}[tbhp]
  \centerline{\includegraphics[width=0.7\textwidth]{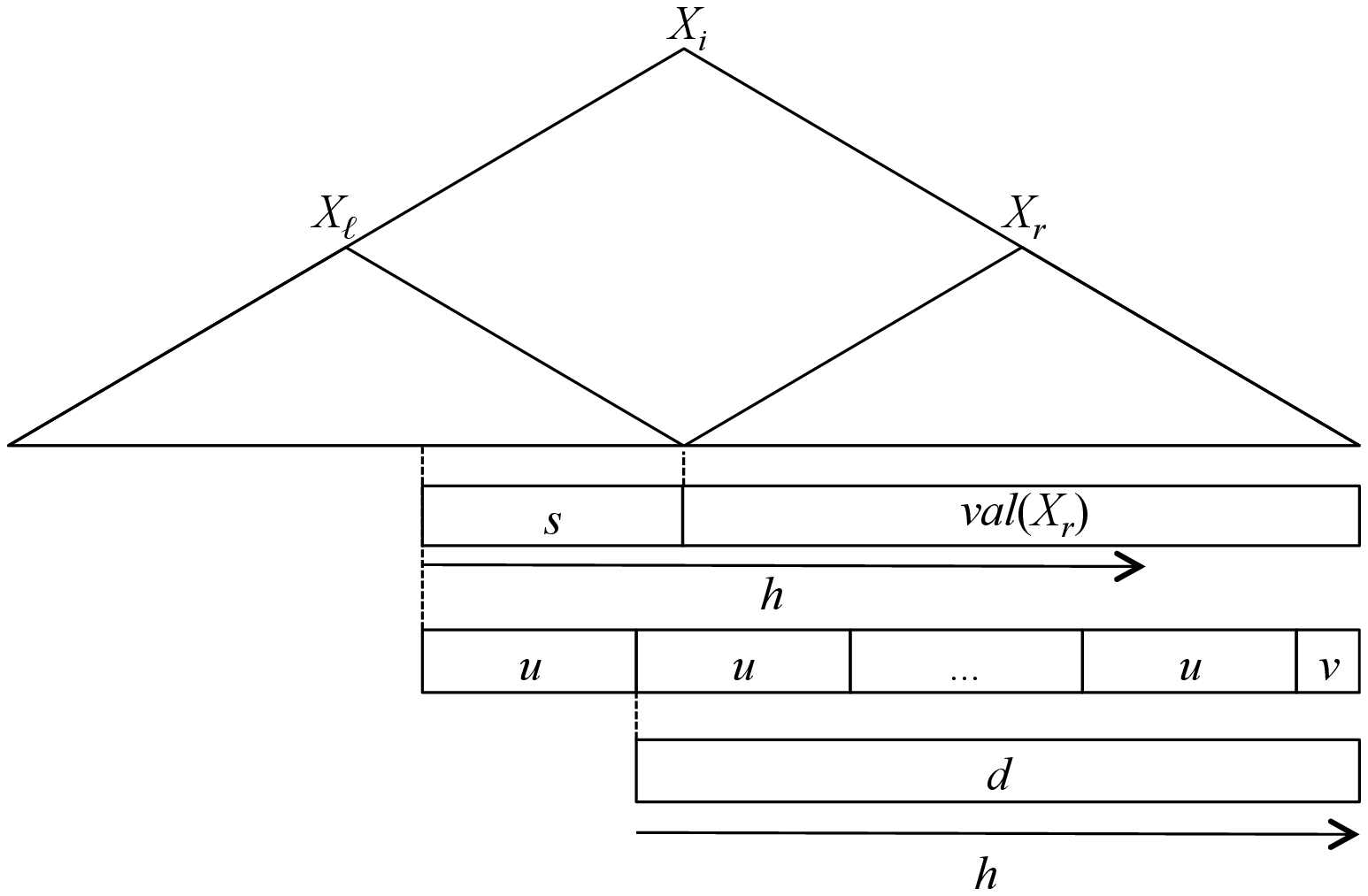}}
  \caption{Lemma~\ref{lem:kouho_O(n^3)}: Case where $h = |d|$ and $|s \cdot \derive(X_r)| \leq 2|d|$. Since $s \cdot \derive(X_r) = u^kv$ and $d = u^{k-1}v$, $d$ is deleted from $D_i$ and $s \cdot \derive(X_r)$ is added to $D_i$.
  }
  \label{fig:kouho_3_1}
\end{figure}

\item If $|s \cdot \derive(X_r)| > 2|d|$,
then both $d$ and $s \cdot \derive(X_r)$ may be a member of $\kouho(X_i)$.
Thus we add $s \cdot \derive(X_r)$ to $D_i$, and update $d \leftarrow s \cdot \derive(X_r)$.
See also Fig.~\ref{fig:kouho_3_2}.
\begin{figure}[tbhp]
  \centerline{\includegraphics[width=0.7\textwidth]{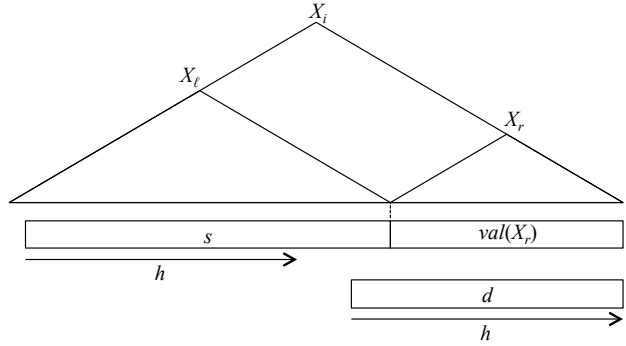}}
  \caption{Lemma~\ref{lem:kouho_O(n^3)}: Case where $h = |d|$ and $|s \cdot \derive(X_r)| > 2|d|$. We add $s \cdot \derive(X_r)$ to $D_i$, and $s \cdot \derive(X_r)$ becomes the longest member of $D_i$.
  }
  \label{fig:kouho_3_2}
\end{figure}

\end{itemize}
\end{itemize}
We represent the strings in $\kouho(X_\ell)$,
$\kouho(X_r)$, $\kouho(X_i)$, and $D_i$ by their lengths.
Given sorted lists of $\kouho(X_\ell)$ and $\kouho(X_r)$,
the above algorithm computes a sorted list for $D_i$,
and it follows from Lemma~\ref{lem:kouho_twice_longer}
that the number of elements in $D_i$ is always $O(\log N)$.
Thus all the above operations on $D_i$
can be conducted in $O(\log N)$ time in each step.

We now show how to efficiently compute $h = \lcp(s \cdot \derive(X_r), d)$,
for any $s \in \kouho(X_\ell)$.
Let $z$ be the longest string in $\kouho(X_\ell)$,
and consider to process any string $s \in \kouho(X_\ell)$.
Since $s$ is a prefix of $z$ by Lemma~\ref{lem:shorter_kouho_prefix},
we can compute $\lcp(s \cdot \derive(X_r), d)$ as follows:
\[
 \lcp(s \cdot \derive(X_r), d) = 
 \begin{cases}
  \lcp(z, d) & \mbox{if } \lcp(z, d) < |s|, \\
  |s| + \lcp(X_r, d[|s|+1..|d|]) & \mbox{if } \lcp(z, d) \geq |s|.
 \end{cases}
\]
To compute the above lcp values using the $\FirstMismatch$ function,
for each variable $X_i$ of $\mathcal{S}$
we create a new production $X_{n+i} = X_i X_i$,
and hence the number of variables increases to $2n$.
In addition, we construct a new SLP of size $O(n)$
that derives $z$ in $O(n)$ time using Lemma~\ref{lem:substring_SLP}.
Let $Z$ be the variable such that $\derive(Z) = z$.
It holds that
\begin{eqnarray*}
\lcp(z, d) & = & \min\{\lcp(Z, X_{n+i}[|X_i|-|d|+1..|X_{n+i}|]), |d|\} \mbox{ and} \\
\lcp(X_r, d[|s|+1..|d|]) & = & \min\{\lcp(X_r, X_{n+r}[|X_r|-|d|+|s|+1..|X_{n+r}|]), |d|-|s|\}.
\end{eqnarray*}
See also Fig.~\ref{fig:kouho_lcp_zd} and Fig.~\ref{fig:kouho_lcp_xr}.
\begin{figure}[tb]
  \centerline{\includegraphics[width=0.7\textwidth]{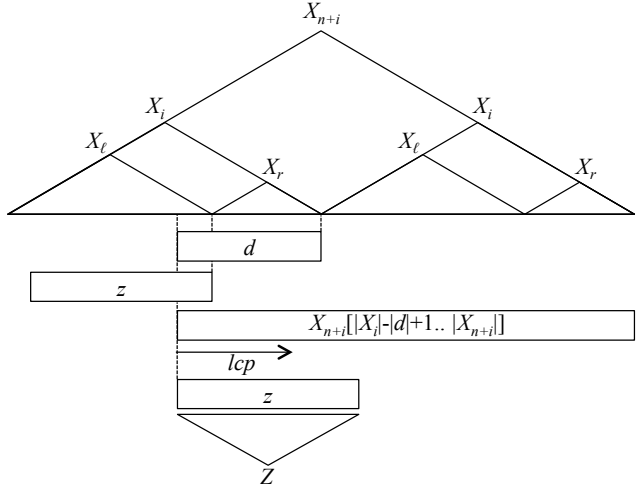}}
  \caption{Lemma~\ref{lem:kouho_O(n^3)}: $\lcp(z, d) = \min\{\lcp(Z, X_{n+i}[|X_i|-|d|+1..|X_{n+i}|]), |d|\}$.
  }
  \label{fig:kouho_lcp_zd}
\end{figure}

\begin{figure}[tb]
  \centerline{\includegraphics[width=0.7\textwidth]{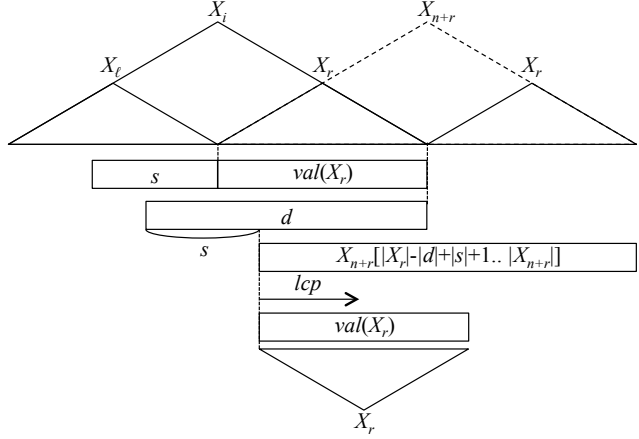}}
  \caption{Lemma~\ref{lem:kouho_O(n^3)}: $\lcp(X_r, d[|s|+1..|d|]) = \min\{\lcp(X_r, X_{n+r}[|X_r|-|d|+|s|+1..|X_{n+r}|]), |d|-|s|\}$.
  }
  \label{fig:kouho_lcp_xr}
\end{figure}
By using Lemma~\ref{lem:lcp_slp}, 
we preprocess, in $O(n^3)$ time and $O(n^2)$ space, the SLP consisting of these variables
so that the query $\FirstMismatch(X_i, X_j, k)$ 
for answering $\lcp(X_i[k..|X_i|], X_j)$ 
is supported in $O(n^2)$ time.
Therefore $\lcp(s \cdot \derive(X_r), d)$ can be computed 
in $O(n^2)$ time for each $s \in \kouho(X_\ell)$.
Since there exist $O(\log N)$ elements in $\kouho(X_\ell)$,
we can compute $\kouho(X_i)$ in $O(n^3 + n^2 \log N) = O(n^3)$ time.
The total space complexity is $O(n^2)$.
\qed
\end{proof}
Since there are $n$ productions in a given SLP,
using Lemma~\ref{lem:kouho_O(n^3)} we can compute $\kouho(X_n)$
for the last variable $X_n$ in a total of $O(n^4)$ time.
The main result of this paper follows.
\begin{theorem}
Given an SLP $\mathcal{S}$ of size $n$ representing a string $w$,
we can compute $\LF(w)$
in $O(n^4 + m n^3 h)$ time and $O(n^2)$ space,
where $m$ is the number of factors in $\LF(w)$ and $h$ is the height of the derivation tree of $\mathcal{S}$.
\end{theorem}
\begin{proof}
Let $\LF(w) = \ell_1^{p_1} \cdots \ell_m^{p_m}$.
First, using Lemma~\ref{lem:kouho_O(n^3)} we compute $\kouho$ for all variables in $\mathcal{S}$ in $O(n^4)$ time.
Next we will compute the Lyndon factors from right to left.
Suppose that we have already computed $\ell_{j+1}^{p_{j+1}} \cdots \ell_m^{p_m}$,
and we are computing the $j$th Lyndon factor $\ell_j^{p_j}$.
Using Lemma~\ref{lem:substring_SLP},
we construct in $O(n)$ time 
a new SLP of size $O(n)$ describing $w[1..|w|-\sum_{k=j+1}^{m} p_k|\ell_k|]$,
which is the prefix of $w$ obtained by removing 
the suffix $\ell_{j+1}^{p_{j+1}} \cdots \ell_m^{p_m}$ from $w$.
Here we note that the new SLP actually has $O(h)$ new variables 
since $w[1..|w|-\sum_{k=j+1}^{m} p_k|\ell_k|]$ can be represented by a sequence of $O(h)$ variables in $\mathcal{S}$.
Let $Y$ be the last variable of the new SLP.
Since $\kouho$ for all variables in $\mathcal{S}$ have already been computed,
it is enough to compute $\kouho$ for $O(h)$ new variables.
Hence using Lemma~\ref{lem:kouho_O(n^3)},
we compute a sorted list of $\kouho(Y) = \kouho(w[1..|w|-\sum_{k=j+1}^{m} p_k|\ell_k|])$ in a total of $O(n^3 h)$ time.
It follows from Lemma~\ref{lem:LF_shortest_kouho} that
the shortest element of $\kouho(Y)$ is $\ell_j^{p_j}$, the $j$th Lyndon factor of $w$.
Note that each string in $\kouho(Y)$ is represented by its length,
and so far we only know the total length $p_{j}|\ell_j|$ of the $j$th Lyndon factor.
Since $\ell_j$ is border free,
$|\ell_j|$ is the shortest period of $\ell_j^{p_j}$.
We construct a new SLP of size $O(n)$ describing $\ell_j^{p_j}$,
and compute $|\ell_j|$ in $O(n^3 \log N)$ time using Lemma~\ref{lem:SLP_period}.
We repeat the above procedure $m$ times,
and hence $\LF(w)$ can be computed in a total of $O(n^4 + m(n^3 h + n^3 \log N)) = O(n^4 + m n^3 h)$ time.
To compute each Lyndon factor of $\LF(w)$,
we need $O(n^2)$ space for Lemma~\ref{lem:SLP_period} and Lemma~\ref{lem:kouho_O(n^3)}.
Since $\kouho(X_i)$ for each variable $X_i$ requires $O(\log N)$ space,
the total space complexity is $O(n^2 + n \log N) = O(n^2)$.
\qed
\end{proof}

\section{Conclusions and open problem}
Lyndon words and Lyndon factorization are important concepts of 
combinatorics on words, with various applications.
Given a string in terms of an SLP of size $n$,
we showed how
to compute the Lyndon factorization of the string in $O(n^4 + mn^3h)$ time
using $O(n^2)$ space,
where $m$ is the size of the Lyndon factorization and $h$ is the height of the SLP.
Since the decompressed string length $N$ can be exponential w.r.t. $n, m$ and $h$,
our algorithm can be useful for highly compressive strings.

An interesting open problem is to compute the Lyndon factorization
from a given LZ78 encoding~\cite{LZ78}.
Each LZ78 factor is a concatenation of the longest previous factor
and a single character.
Hence, it can be seen as a special class of SLPs,
and this property would lead us to a much simpler and/or more efficient solution to the problem.
Noting the number $s$ of the LZ78 factors is $\Omega(\sqrt{N})$,
a question is whether we can solve this problem in $o(s^2) + O(m)$ time.

\bibliographystyle{splncs03}
\bibliography{ref}

\end{document}